\documentclass{llncs}
\usepackage{indentfirst,amssymb,amsmath,latexsym,graphicx}

\def\O{{\mathcal O}}
\def\I{{\mathcal I}}
\def\R{{\tt R}}
\def\S{{\tt S}}
\def\Int{{\tt L }}
\def\Rule#1#2#3{#1 \ \frac{#2}{#3}}
\def\bigsqcap{\rotatebox[origin=c]{180}{$\bigsqcup$}} 

\def\Si{\Sigma}

\def\R{{\sf{R} }}

\def\I{{\mathcal{I}}}
\def\Nc{{\sf{N_c}}}

\def\Nr{{\sf{N_r}}}

\def\EL{{\mathcal{EL}}}

\def\dleq{\sqsubseteq}

\def\bigsqcap{\rotatebox[origin=c]{180}{$\bigsqcup$}} 
\def\dcap{\sqcap}
\def\dcup{\sqcup}
\def\ex#1{\exists #1}
\def\sig{{\tt sig}\,}

\def\lldots{,\ldots ,}
\def\QED{$\square$\medskip}

\begin{document}

\title{DeFind: A Protege Plugin for Computing Concept Definitions in $\cal{EL}$ Ontologies\thanks{The research was supported by the Russian Science Foundation, grant No. 17-11-01166. Implementation was supported by the Siberian Branch of the Russian Academy of Sciences (Block 36.1 of the Program for Basic Scientific Research II.1)}}

\author{Denis Ponomaryov\inst{1,2} \and Stepan Yakovenko\inst{2}}


\institute{Institute of Mathematics, Novosibirsk State University  \and Institute of Informatics Systems, Novosibirsk State University\newline
\email{ponom@iis.nsk.su, stiv.yakovenko@gmail.com}}


\maketitle

\begin{abstract}
We introduce an extension to the Prot\'eg\'e ontology editor, which allows for discovering concept definitions, which are not  explicitly present in axioms, but are logically implied by an ontology. The plugin supports ontologies formulated in the Description Logic $\cal{EL}$, which underpins the OWL 2 EL profile of the Web Ontology Language and despite its limited expressiveness captures most of the biomedical ontologies published on the Web. The developed tool allows to verify whether a concept can be defined using a vocabulary of interest specified by a user. In particular, it allows to decide whether some vocabulary items can be omitted in a formulation of a complex concept. The corresponding definitions are presented to the user and are provided with explanations generated by an ontology reasoner. 
\end{abstract}

\section{Introduction}\label{Sect:Introduction} \vspace{-0.25cm}
The development and application of large terminological systems pose new challenges for ontology engineering and automated reasoning tools. A question which becomes evidently important in the context of large ontologies is whether their content and logical consequences are easy to comprehend. To address this problem, a number of visualization and explanation tools has been integrated into ontology editing environments such as, e.g., visualization and explanation services implemented for the Prot\'eg\'e ontology editor. It is common that an expert exploring an ontology encounters concepts, which are formulated in a vocabulary she is not familiar with. More generally, only a certain part of an ontology vocabulary may be familiar to an expert, while the remaining part may be not. When a formulation of some concept employs vocabulary items unknown to an expert, typically she would like to know whether this concept can be reformulated in the familiar vocabulary. To give a simplified example, suppose an ontology of cuisines contains axiom $Dumplings \sqcap Entree \sqsubseteq Gnocci$ (stating that $Dumplings$ being $Entree$ are $Gnocci$) together with concept inclusions $Gnocci \sqsubseteq Dumplings$ and $Dumplings\sqsubseteq Entree$. Assume an expert is familiar with concepts $Entree$ and $Gnocci$ and she encounters the concept $Dumplings \sqcap Entree$ mentioned in the ontology. First, one may notice that this concept can be simplified, i.e., it is equivalent to $Dumplings$ wrt the ontology, due to the inclusions above. Second, the ontology entails that $Dumplings$ and $Gnocci$ are equivalent, thus, the original concept can be reformulated as $Gnocci$ (or as $Gnocci \sqcap Entree$) in the vocabulary known to the user.  

In this paper, we introduce DeFind, an extension to the Prot\'eg\'e ontology editor, which allows to find concept definitions in a user specified vocabulary consisting of concept and relations names from an input ontology. In particular, it allows to verify whether some concept or relation names can be omitted in a formulation of a (complex) concept. To the best of our knowledge, this is the first implementation of these reasoning services as a part of an ontology editing tool. Specifically, for a given ontology $\O$, vocabulary $\Sigma$, and a concept $C$ of interest, DeFind computes definitions of $C$ wrt $\O$ in $\Sigma$, i.e. concepts $\{D_1\lldots D_n\}$ (whenever they exist) such that $D_i$ contains symbols only from $\Sigma$ and $\O\models C\equiv D_i$, for all $i=1\lldots n$. DeFind supports ontologies formulated in the Description Logic $\cal{EL}$, which allows for building concepts using conjunction and existential restriction, and includes built-in concepts such as $\bot$ and $\top$. For example, it is possible to state in $\EL$ that some concepts are disjoint (e.g., $Entree \sqcap Dessert \sqsubseteq \bot$), specify subsumption relationship between concepts ($Dumplings \sqcap Entree \sqsubseteq Gnocci$) or domains of relations ($\ex{hasIngredient}.\top\sqsubseteq Food$), and use existential restriction to specify relationships of other kinds (e.g., $Salad \sqsubseteq \ex{hasDressing}.\top$, $\ex{hasIngredient}.Meat\sqsubseteq NonVegeterianFood$). The expressive features of $\EL$ although limited (e.g., negation/disjunction of concepts is not allowed), are sufficient to capture a great variety of ontologies. Most of the biomedical ontologies published on the Web fall within the formalism of $\cal{EL}$, which underpins the OWL 2 EL profile of the Web Ontology Language. 


\section{Techniques}\label{Sect:Techniques} \vspace{-0.25cm}
In the Description Logic $\cal{EL}$, concepts are built using a countably infinite alphabet of \emph{roles names} $\Nr$ and \emph{concept names} $\Nc$, with two distinguished concepts $\bot, \top\in\Nc$. The notion of \emph{concept} is defined inductively: any element of $\Nc$ is a concept and if $C, D$ are concepts and $r\in\Nr$ is a role then $C\sqcap D$ and $\ex{r}.C$ are concepts. A \emph{concept inclusion} is an expression of the form $C\sqsubseteq D$, where $C, D$ are concepts. A \emph{concept equivalence} is a expression $C\equiv D$. An \emph{ontology} is a finite set of concept inclusions and equivalences (called \emph{axioms}). 
A \emph{signature} is a subset of $\Nr\cup\Nc\setminus\{\bot, \top\}$. The signature of a concept $C$, denoted as $\sig(C)$, is the set of role and concept names (excluding $\bot$ and $\top$), which occur in $C$. The signature of a concept inclusion or ontology is defined similarly. 
Semantics is defined by using the notion of interpretation, which is a pair $\I= \langle \Delta, \cdot^\I\rangle$, where $\Delta$ is a universe and $\cdot^\I$ is a function which maps every concept name from $\Nc$ to a subset of $\Delta$ such that $(\bot)^\I=\varnothing$,  $(\top)^\I=\Delta$, and every role name to a subset of $\Delta\times \Delta$. This function is extended to arbitrary concepts by setting $(C\sqcap D)^\I=C^\I\cap D^\I$ and $(\ex{r}.C)^\I=\{x\in \Delta \mid \exists y \ \langle x,y\rangle\in r^\I \ \text{and} \ \ y\in C^\I\}$. An interpretation $\I$ is a model of a concept inclusion $C\sqsubseteq D$ (written as $\I\models C\sqsubseteq D$) if $C^\I\subseteq D^\I$. Similarly, $\I\models C\equiv D$ if $C^\I=D^\I$. An interpretation is a model of an ontology $\O$ if it is a model of every axiom of $\O$. An ontology $\O$ \emph{entails} a concept inclusion $C\dleq D$ (written as $\O\models C\dleq D$) if $\I\models\O$ yields $\I\models C\sqsubseteq D$, for any interpretation $\I$.

Let $\Sigma$ be a signature. We say that a concept $C$ is \emph{$\Sigma$-definable} wrt an ontology $\O$ if $\O\models C\equiv D$, where $\sig(D)\subseteq\Sigma$. The concept $D$ is called $\Sigma$-\emph{definition} of $C$  (wrt ontology $\O$). For example, the concept $Dumplings \sqcap Entree$ is $\Sigma$-definable wrt the ontology from the introduction, where $\Sigma=\{Gnocci\}$. It is known that in $\EL$ a concept $C$ is $\Sigma$-definable wrt an ontology $\O$ iff it holds $\O\cup\O^*\models C\dleq C^*$, where $\O^*$ and $C^*$ are ``copies'' of $\O$ and $C$, respectively, obtained by an injective renaming of all non-$\Sigma$-symbols into ``fresh'' ones, not occurring in $\O$ and $C$. Indeed, if $\O\models C\equiv D$, for some concept $D$, with $\sig(D)\subseteq\Sigma$, then it holds $\O\models C\dleq D$ and $\O^*\models D\dleq C^*$, which means $\O\cup\O^*\models C\dleq C^*$. The converse can be proved constructively by computing the corresponding concept $D$ from a proof of the inclusion $C\dleq C^*$ from $\O\cup\O^*$; we provide a justification below. Various proof systems have been proposed for the Description Logic $\cal{EL}$ in the literature. For example, it follows from the results in \cite{IncredibleELK} that the following set of inference rules is sound and complete for entailment of concept inclusions $C\dleq D$ from an ontology $\O$, where $C$ and $D$ occur in the axioms of $\O$:

\vspace{-0.8cm}

\begin{figure}\label{Fig:ELRules}
$$\Rule{\R_0}{}{C\dleq C} \hspace{0.3cm} \ \Rule{\R_\top}{}{C\dleq \top}  \hspace{0.3cm} \ \Rule{\R_\bot}{}{\bot\dleq C}  \hspace{0.5cm} \ \Rule{\R_\dleq}{C\dleq E}{C\dleq F} \ E\bowtie F\in\O, \ \bowtie \ \in\{\dleq, \equiv\}$$
$$\Rule{\R_\dcap^-}{C\dleq E_1\dcap E_2}{C\dleq E_1 \ C\dleq E_2} \hspace{0.6cm} \ \Rule{\R_\dcap^+}{C\dleq E_1 \ C\dleq E_2}{C\dleq E_1\dcap E_2} \ E_1\dcap E_2 \ \text{occurs in} \ \O $$

$$\ \Rule{\R_\exists^\bot}{C\dleq\ex{r}.E \ E\dleq\bot}{C\dleq\bot} \hspace{0.5cm} \ \Rule{\R_\exists}{C\dleq\ex{r}.E \ E\dleq F}{C\dleq \ex{r}.F}  \ \ex{r}.F \ \text{occurs in} \ \O $$
\caption{Basic inference rules for reasoning in $\cal{EL}$}
\end{figure}

\vspace{-0.4cm}

An \emph{inference} is a triple $\langle$rule name, premises, conclusion$\rangle$ of the form $\langle \R, \{\varphi_1,$ $. . ,\varphi_n\}, \psi \rangle$, where $n\geqslant 0$ and $\psi$ is obtained by an inference rule $\R$ with premises $\varphi_1,\lldots \varphi_n$. A \emph{proof} of a concept inclusion $\psi$ (from an ontology $\O$) is a set of inferences $\{ \iota_1\lldots\iota_n \}$, where $n\geqslant 1$, such that $\psi$ is the conclusion of $\iota_n$ and for all $k=1\lldots n$ and any axiom $\varphi$, if $\varphi$ is a premise of $\iota_k$, then there is a unique $j<k$ such that $\varphi$ is the conclusion of $\iota_j$ and if $\varphi$ is the conclusion of $\iota_k$ and $k\neq n$ then there is $l>k$ such that $\varphi$ is a premise of $\iota_l$.
Observe that due to the injective renaming of non $\Sigma$-symbols it holds $\sig(\O)\cap\sig(\O^*)\subseteq\Sigma$ and whenever a concept $C$ occurs in ontology $\O$, we have $\sig(C)\subseteq\sig(\O)$ and hence, $\sig(C^*)\subseteq\sig(\O^*)$. If there is a concept $D$ such that $\sig(D)\subseteq\Sigma$ and $\O\cup\O^*\models \{C\dleq D, \ D \dleq C^*\}$ (in which case $D$ is called \emph{interpolant} for $C\dleq C^*$) then due to the renaming it holds that $\O\models\{ C\dleq D, D\dleq C\}$ and thus, $\O\models C\equiv D$, i.e., $D$ is a definition of concept $C$ in signature $\Sigma$.

\begin{theorem}\label{Thm:ComputingInterpolant}
Let $\Si$ be a signature, $\O_1, \O_2$ ontologies and $C_1, C_2$ concepts such that $\sig(\O_1)\cap\sig(\O_2)\subseteq\Sigma$ and $C_i$ occurs in $\O_i$, for $i=1,2$. If $\O_1\cup\O_2\models C_1\dleq C_2$ then there exists a concept $D$ such that $\sig(D)\subseteq\Sigma$ and $\O_1\cup\O_2\models\{C_1\dleq D, \ D\dleq C_2\}$. 
\end{theorem}
\begin{proof}
We note the following property (referred to as $\star$), which immediately follows from the definition of the inference rules: if there is a proof of a concept inclusion $C \dleq E$ from an ontology $\O$ then either it is obtained by one of the rules $\R_\top, \R_\bot, \R^\bot_\ex$, or $E$ occurs in $C$ or $\O$. 
We use induction on the number of inferences in a proof $\langle \iota_i\lldots\iota_n\rangle$ of a concept inclusion $\varphi=C_1\dleq C_2$ from $\O_1\cup\O_2$. For $n=1$, if $\varphi$ is obtained by $\R_0$ or $\R_\top$ then $\sig(C_2)\subseteq\Sigma$ and thus, $C_2$ is an interpolant for $\varphi$. If it is obtained by $\R_\bot$, then $\bot$ is an interpolant.
For the induction step, if $\varphi$ is obtained by $\R_\ex^\bot$ then $\bot$ is an interpolant. If the rule is $\R_\sqcap^+$ then by the induction assumption there is an interpolant $D_i$ for each premise $C\dleq E_i$, $i=1,2$, and hence, $D_1\dcap D_2$ is an interpolant for $\varphi$. If the rule is $\R_\dleq$ then $E\bowtie F\in\O_1\cup\O_2$, for some $i=1,2$, where $\bowtie \ \in\{\dleq, \equiv\}$. If $i=1$ then it follows from $\sig(F)\subseteq\sig(\O_2)$ that $\sig(F)\subseteq\Sigma$ and thus, $F$ is an interpolant for $\varphi$. If $i=2$ then $\sig(E)\subseteq\sig(\O_2)$, thus by the induction assumption there is an interpolant for $C\dleq E$, which is an interpolant for $\varphi$. 
If the rule is $\R_\dcap^-$, consider its premise $C\dleq E_1\dcap E_2$; w.l.o.g we assume that $\varphi=C\dleq E_1$. It follows from ($\star$) that $\sig(E_1\dcap E_2)\subseteq\sig(\O_i)$, for some $i=1,2$. If $i=2$, then there is an interpolant for $C\dleq E_1\dcap E_2$, which is an interpolant for $\varphi$. If $i=1$ then $\sig(E_1)\subseteq\sig(\O_2)$ yields $\sig(E_1)\subseteq\Sigma$ and hence, $E_1$ is an interpolant for $\varphi$.  Finally, if the rule is $\R_\exists$ then it follows from ($\star$) that $\sig(\ex{r}.E)\subseteq\sig(\O_i)$, for some $i=1,2$. If $i=2$ then there is an interpolant for $C\dleq\ex{r}.E$, which is an interpolant for $\varphi$. If $i=1$ then it follows from the condition $\sig(\ex{r}.F)\subseteq\sig(\O_2)$ that $r\in\Sigma$ and by the induction assumption there is an interpolant $D$ for $E\dleq F$. Then $\ex{r}.D$ is an interpolant for $\varphi$. \QED
\end{proof}

\vspace{-0.2cm}

 The proof of the theorem gives an idea of a recursive algorithm for computing interpolants and hence, concept definitions, by traversing proofs of $C\dleq C^*$ from the union $\O\dcup\O^*$. In general, there can exist several proofs and each of them can yield a different definition of the concept $C$. Although the idea is simple, its implementation is not straightforward. First, it requires an ontology reasoner, which is not only able to decide entailment of concept inclusions from an ontology, but supports proof tracing, i.e., provides proofs as certificates for entailment. Second, the method presented by Theorem \ref{Thm:ComputingInterpolant} has to rely on a proof system implemented by a reasoner. DeFind employs ELK \cite{IncredibleELK}, which is a highly optimized reasoner for the Description Logic $\cal{EL}$ and its extensions. The add-ons to ELK, which provide proof tracing and explanation services, employ the set of inference rules given in Figure \ref{Fig:ELKReasonerRules}. Showing a direct analogue of Theorem \ref{Thm:ComputingInterpolant} for this proof system is not straightforward, since (in contrast to the rules from Figure 1) there exist proofs, from which an interpolant can not be directly computed. However, it can be shown that there always exists (at least a single) proof, which is appropriate for computing interpolants and hence, concept definitions.
 
 \smallskip


\begin{figure}\label{Fig:ELKReasonerRules}
$$\Rule{\S_0}{}{C\dleq C} \hspace{0.7cm}  \Rule{\S_\top}{}{C\dleq \top} \hspace{0.7cm} \Rule{\S_\bot}{}{\bot\dleq C}  \hspace{0.7cm} \Rule{\S_{ax}}{}{C\dleq E} \ C\dleq E\in\O$$

$$\Rule{\S_\dleq}{C_0\dleq C_1\ldots C_{n-1}\dleq C_n}{C_0\dleq C_n} \hspace{0.4cm} \Rule{\S_\equiv}{}{C_j\dleq C_k} \ C_i\equiv C_{i+1}\in\O, \ \substack{0\leqslant i < n}, \ \substack{ 0\leqslant j,k \leqslant n}$$

$$\Rule{\S_\dcap^-}{}{C_1\sqcap\ldots \sqcap C_n\dleq C_i} \hspace{0.6cm} \Rule{\S_\dcap^+}{C\dleq E_1 \ \ldots \  C\dleq E_n}{C\dleq E_1\sqcap\ldots\sqcap E_n} $$

$$\ \Rule{S_\exists^\bot}{}{\ex{r}.\bot\dleq\bot} \hspace{0.7cm} \Rule{\S_\exists}{C\dleq E}{\ex{r}.C\dleq \ex{r}.E} $$
\caption{Inference rules used by the proof tracing service for ELK reasoner}

\end{figure}

DeFind implements a recursive procedure, which traverses proofs obtained by using the tracing service for ELK reasoner. Proofs are built for the entailment $\O\cup\O^*\models C\dleq C^*$, where $\O^*$ and $C^*$ are ``copies'' (of a given ontology $\O$ and concept $C$) constructed wrt a specified signature $\Sigma$. The procedure applies recursively the rules below to compute a label expression $\Int(\varphi)$ for every concept inclusion $\varphi$ appearing as the conclusion of an inference and outputs a label computed for $C\dleq C^*$. Whenever there are several proofs for the same $\varphi$, each of them is traversed. A \emph{label} is a concept formulated in an extension of $\EL$ with disjunction $\dcup$ (denoted as $\EL_\dcup$) and a distinguished ``empty'' concept $\epsilon$, for which the following holds: $\epsilon\dcap D=\epsilon$, $\epsilon\dcup D=D$, $\ex{r}.\epsilon=\epsilon$, and $\ex{r}.(D\dcup E) = \ex{r}.D\dcup\ex{r}.E$, for any concepts $D,E$ and role $r$. By using the latter equation, the notion of Disjunctive Normal Form of a $\EL_\dcup$-concept is naturally defined. We say that a label expression is \emph{empty} if it equals $\epsilon$. Initially the label of every conclusion appearing in a proof is assumed to be empty. Each rule below is provided with a name and gives a label to the conclusion $\varphi$ of an inference being visited during proof traversal, depending on the type of the inference rule and the labels of its premises. We use the notations from Figure \ref{Fig:ELKReasonerRules} for the premises of each inference rule mentioned below:

\begin{description}
\item[($\Int_\bot$)] rule is $\S_\bot$ $\Rightarrow$ $\Int(\varphi) := \Int(\varphi)\dcup \bot$
\item[($\Int_\Sigma$)] rule is one of $\S_0, \S_\top, \S_{ax}, \S_\equiv, \S_\dcap^-, \S_\ex^\bot$, $\varphi=C\dleq E$, and $\sig(E)\subseteq\Sigma$ $\Rightarrow$ $\Int(\varphi) := \Int(\varphi)\dcup E$
\item[($\Int_\dleq$)] rule is $\S_\dleq$ $\Rightarrow$ $\Int(\varphi) := \Int(\varphi) \ \dcup \ \bigsqcup_{i=1\lldots n} \  \Int(C_{i-1}\dleq C_i)$
\item[($\Int_\sqcap^+$)] rule is $\S_\sqcap^+$  $\Rightarrow$ $\Int(\varphi) := \Int(\varphi) \ \dcup \  \bigsqcap_{i=1\lldots n} \ \Int(C\dleq E_i)$
\item[($\Int_\exists$)] rule is $\S_\exists$ and $r\in\Sigma$ $\Rightarrow$ $\Int(\varphi) := \Int(\varphi) \dcup \ex{r}.\Int(C\dleq E)$
\end{description}

\begin{theorem}\label{Thm:ComputingLabel}
Let $\Si$ be a signature, $\O$ ontology, and $C$ a concept occurring in $\O$. The concept $C$ is $\Sigma$-definable wrt $\O$ iff the procedure returns a non-empty $\EL_\dcup$-expression $\Int(C\dleq C^*)$. Every conjunct from the Disjunctive Normal Form of $\Int(C\dleq C^*)$ is a $\EL$-concept and it is a $\Sigma$-definition of $C$ wrt $\O$.
\end{theorem}

\section{Complexity of Definitions}\label{Sect:Complexity} \vspace{-0.25cm}

We now estimate the size and the number of definitions computed by DeFind. The above mentioned procedure recursively computes labels for conclusions of inferences appearing in proofs. Under the assumption that every proof is computed in polynomial time in the size of an input ontology $\O$ (which is true when using the ELK reasoner), the size of every proof is polynomially bounded by the size of $\O$. For any conclusion, the label is computed once, when traversing a single proof. The rules ($\Int_\bot$), ($\Int_\Sigma$) give a label concept from a conclusion of some inference and hence, its size is bounded by the size of $\O$. Every other rule of the procedure gives a label, which is obtained as a combination (under $\dcup, \ \dcap$, or $\ex{r}.$, with $r\in\sig(\O)$) of $k$ expressions computed at the previous steps of the recursion, where $k\geqslant 1$ is the number of premises of an inference rule and hence, $k$ is bounded by the size of the proof. In particular, the rules $\Int_\dleq$, $\Int_\dcap^+$ give a label, which is at most $k$ times longer than the labels of the premises of $\S_\dleq$, $\S_\dcap^+$. Therefore, the size of any label computed from a single proof is at most exponential in the size of the ontology. Since every proof is polynomially bounded, there are at most exponentially many proofs for any conclusion and hence, the output of the procedure has at most exponential size. By Theorem \ref{Thm:ComputingLabel}, every conjunct from the DNF of the output label is a definition and therefore, the upper bound on the number of definitions computed by DeFind is double exponential in the size of the ontology. 
To estimate the size of the computed definitions one can w.l.o.g. assume that every rule of the procedure gives a label in DNF and show the following by induction. Every conjunct from a label obtained by some rule of the procedure is either a concept from a conclusion of an inference from a proof, or a combination (under $\dcup, \ \dcap$, or $\ex{r}.$) of conjuncts from the labels computed at the previous steps. Therefore, the size of every definition is at most exponential in the size of the input ontology. 
We now show that these bound are tight by giving an example of an ontology $\O$, concept $C$, and a signature $\Sigma$ such that there is a double exponential number of shortest $\Sigma$-definitions of $C$ wrt $\O$, where every definition has size exponential in the size of $\O$. Let $\Sigma=\{r, s, D_1, D_2\}$, $C=A_0$, and for $n\geqslant 1$, let $\O$ consist of  axioms: $A_0\equiv \ex{r}.A_1\dcap\ex{s}.A_1 \lldots$  $ A_{n-1}\equiv \ex{r}.A_n\dcap\ex{s}.A_n$, $\ A_n\equiv D_1$, $\ A_n\equiv D_2$. Observe that $\O\models A_n\equiv D_1\dcup D_2$, hence, $\O\models A_{n-1}\equiv \ex{r}.(D_1\dcup D_2) \dcap \ex{s}.(D_1\dcup D_2)$. Converting this expression into DNF gives four $\Sigma$-definitions of $A_{n-1}$ wrt $\O$, i.e., concepts $\ex{r}.D_i \dcap \ex{s}.D_j$, where $i,j=1,2$. Similarly, $A_0$ is equivalent wrt $\O$ to a $\EL_\dcup$-concept, which has $2^n$ occurrences of $D_1\dcup D_2$. Converting this concept into DNF gives a double exponential number of $\Sigma$-definitions of $A_0$ wrt $\O$ and it can be shown that no shorter $\Sigma$-definitions exist. \vspace{-0.2cm}

\section{Features of DeFind}\label{Sect:Features} \vspace{-0.2cm}

DeFind (available at https://github.com/stiv-yakovenko/defind) requires the ELK reasoner \cite{IncredibleELK} and proof explanation plugin \cite{ExplanationELK} for Prot\'eg\'e. 
The interface of DeFind is given in Figure \ref{Fig:Interface}. The screenshot shows $\Sigma$-definitions of concept $\ex{r}.A_1\dcap\ex{s}.A_1$ computed by Defind wrt ontology $\O$ and signature $\Sigma$ from the previous section for $n=1$. The user can input a complex concept in the Manchester OWL Syntax into the field ``Class expression'' or drag-and-drop a concept name from the class hierarchy of the ontology. To specify concept and role names for the target signature, one can drag-and-drop items from the class or object property hierarchy. One can simultaneously add all role names from the ontology into the target signature by pressing ``Add all Object Properties'' button. Selecting ``doesn't include symbols'' option sets the target signature to consist of all those concept an role names from the ontology, which are not listed in the ``Target signature'' field.  On pressing ``Compute definitions'', the corresponding $\Sigma$-definitions are computed (or a notification is shown if no definition exists). On pressing the question mark at the right-hand side of a definition, an explanation is shown, why it corresponds to the specified class expression: the Prot\'eg\'e proof explanation plugin is called to visualize a proof (see, e.g., Figure 5 in \cite{ExplanationELK}) of the equivalence of the class expression to the concept computed by DeFind. 

\vspace{-0.25cm}

\begin{figure}\label{Fig:Interface}
\makebox{\put(0,0){\includegraphics[height=5cm]{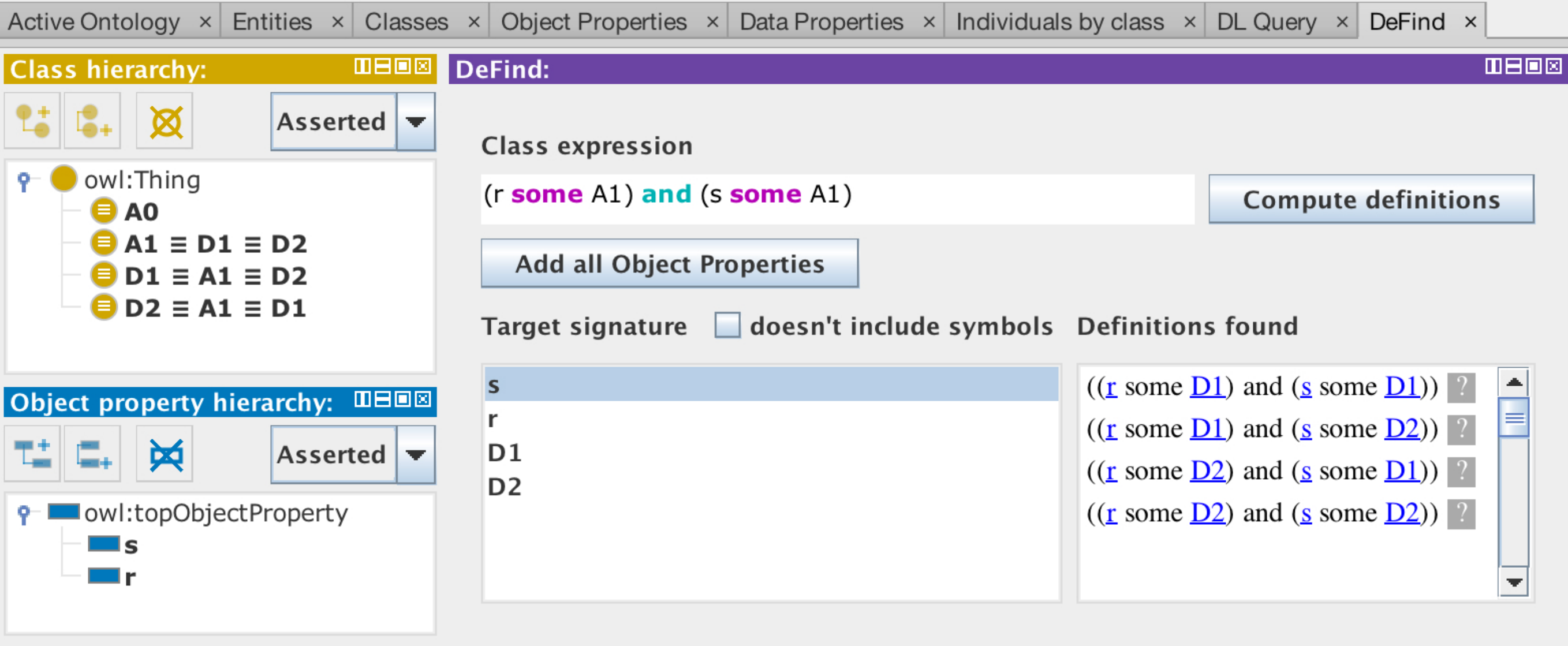}}}
\caption{Interface of DeFind in the Prot\'eg\'e ontology editor}
\end{figure}

\section{Related Work} \vspace{-0.25cm}

In \cite{BethDefinabilityExpressive}, a similar approach is implemented to compute concept definitions from tableaux proofs in the Description Logic $\cal{ALC}$ and its extensions. The authors provide a constructive proof of an analogue of Theorem \ref{Thm:ComputingInterpolant} and derive an algorithm, which computes a definition of at most double exponential size, whenever there exists one in a target signature. This bound is tight and is one exponential larger than the size of the shortest definitions in $\EL$. It is known that general tableaux methods, when applied to $\EL$, provide exponentially longer proofs than the ones obtained in the proof systems from this paper. It should be noted that the algorithm from \cite{BethDefinabilityExpressive} computes a single definition from a proof, whereas DeFind computes several definitions, the number of which is bounded by a double exponential in the size of the ontology. DeFind explores all proofs provided by a reasoner and traversing a single proof can potentially give several definitions for a concept of interest.  In \cite{EnrichedModels}, the authors propose an algorithm, which finds a single concept definition in a signature of interest wrt a normalised $\cal{EL}$-ontology. 
The algorithm computes an algebraic representation of an ontology as a canonical model provided with an information on how each element of the model is obtained. Essentially, it implements the two features, which are required for computing concept definitions, i.e., ontology reasoning and proof tracing. In contrast, DeFind relies on external reasoning and proof tracing services. The question whether it  suffices to compute a single concept definition strongly depends on application. For instance, the method from \cite{BethDefinabilityExpressive} is employed in \cite{QueryReformulation} to compute query reformulations. In \cite{DecompPaper}, checking for existence of a concept definition (or computing a single one) is used for ontology decomposition. In both applications, the quality of the obtained result is strongly related to the form of a computed definition. On the other hand, it is argued in \cite{MDS} that having multiple definitions is helpful in the context of ontology alignment, since having a choice of several (semantically equivalent, but syntactically different) definitions facilitates finding matches between concepts from different ontologies. The authors propose a heuristic approach for computing concept definitions based on examination of concept definition patterns, which typically occur in ontologies. This solution does not employ ontology reasoning and is in general faster, but it can miss some definitions computed by the methods that employ proof tracing.  

\section{Outlook}  \vspace{-0.2cm}

The exponential lower bound for the size and number of shortest concept definitions is shown in this paper by an artificial ontology example. However it demonstrates the natural phenomenon: some concepts can be used as auxiliary ones in ontology to make definitions shorter. If one asks whether there is a definition, which does not contain certain concept or role names, then one may get a positive answer, but the obtained definition may be long and not easy to comprehend. It is to be understood whether a significant increase of definition size can happen in real-world ontologies. Further, it may be important to develop techniques that help to automatically suggest extensions of the target signature, when a blow-up is unavoidable. One of the primary goals of further implementation is to support other language features of OWL 2 EL profile of the Web Ontology Language, e.g., inclusion and composition of roles. These features are particularly interesting, because on one hand they are frequently used in ontologies and on the other hand, Theorem \ref{Thm:ComputingInterpolant} is no longer true for $\cal{EL}$ extended with these features. Finally, there is a room for optimizations for DeFind. If the plugin is used extensively to compute definitions of the same concept wrt different signatures then it makes sense to implement optimizations proposed in \cite{EnrichedModels}, which employ incremental reasoning to reduce computation overhead, when only a small part of the input is changed. 

\bibliographystyle{splncs.bst}
\bibliography{references.bib}

\end{document}